\documentclass[12pt]{article}
\usepackage[small]{titlesec}
\usepackage{amssymb,amsmath,amsthm}
\usepackage{enumerate}
\usepackage{hyperref}
\usepackage{cite}
\usepackage[mathscr]{euscript}
\usepackage[letterpaper,hmargin=2.70cm,vmargin=3.2cm]{geometry}

\usepackage{setspace}
\usepackage{graphicx}
\newtheorem{theorem}{Theorem}[section]
\newtheorem{proposition}[theorem]{Proposition}

\theoremstyle{definition}
\newtheorem{definition}[theorem]{Definition}

\theoremstyle{remark}

\newtheorem*{acknowledgments}{Acknowledgments}

\newcommand{\abs}[1]{\left\lvert #1 \right\rvert}
\newcommand{\norm}[1]{\left\lVert #1 \right\rVert}

\newcommand{\inner}[2]{\left\langle#1,#2\right\rangle}

\newcommand{\cD}{\mathcal{D}}
\newcommand{\cH}{\mathcal{H}}
\newcommand{\cB}{\mathcal{B}}
\newcommand{\cL}{\mathcal{L}}
\newcommand{\R}{{\mathbb R}}
\newcommand{\C}{{\mathbb C}}

\newcommand{\cc}[1]{\overline{#1}}

\newcommand{\ournewclass}{\mathscr{S}(\mathcal{H})}

\renewcommand\hat{\widehat}

\DeclareMathOperator{\im}{im}
\DeclareMathOperator{\dom}{dom}
\DeclareMathOperator{\Ker}{ker}

\DeclareMathOperator{\ran}{ran}
\DeclareMathOperator{\Sp}{spec}
\DeclareMathOperator{\Span}{span}

\DeclareMathOperator{\assoc}{assoc}
\begin{document}
\begin{titlepage}
\title{The spectra of selfadjoint extensions of entire operators
		with deficiency indices (1,1)%
\footnotetext{%
Mathematics Subject Classification(2000):
46E22, 
47A25, 
47B25} 
\footnotetext{%
Keywords: symmetric operators, entire operators, de Branges spaces,
	spectral analysis.}
}
\author{
\textbf{Luis O. Silva}\thanks{Partially supported by CONACYT (M\'exico)
	through grant CB-2008-01-99100}
\\
\small Departamento de M\'{e}todos Matem\'{a}ticos y Num\'{e}ricos
	\\[-1.6mm]
\small Instituto de Investigaciones en Matem\'aticas Aplicadas y
	en Sistemas
	\\[-1.6mm]
\small Universidad Nacional Aut\'onoma de M\'exico
	\\[-1.6mm]
\small C.P. 04510, M\'exico D.F.
	\\[-1.6mm]
\small \texttt{silva@leibniz.iimas.unam.mx}
\\[4mm]
\textbf{Julio H. Toloza}\thanks{Partially supported by CONICET (Argentina)
	through grant PIP 112-200801-01741}
\\
\small CONICET\\[-1.6mm]
\small Centro de Investigaci\'on en Inform\'atica para la Ingenier\'ia
	\\[-1.6mm]
\small Universidad Tecnol\'ogica Nacional -- Facultad Regional C\'ordoba
	\\[-1.6mm]
\small Maestro L\'opez esq.\ Cruz Roja Argentina
	\\[-1.6mm]
\small X5016ZAA C\'{o}rdoba, Argentina
	\\[-1.6mm]
\small \texttt{jtoloza@scdt.frc.utn.edu.ar}}

\date{}
\maketitle
\begin{center}
\begin{minipage}{5in}
\centerline{{\bf Abstract}}
\bigskip
  We give necessary and sufficient conditions for real sequences to be
  the spectra of selfadjoint extensions of an entire operator whose
  domain may be non-dense. For this spectral characterization we use
  de Branges space techniques and a generalization of Krein's
  functional model for simple, regular, closed, symmetric operators
  with deficiency indices (1,1). This is an extension of our previous
  work in which similar results were obtained for densely defined
  operators.
\end{minipage}
\end{center}
\thispagestyle{empty}
\end{titlepage}
\section{Introduction}
The aim of this work is to present a generalization of the spectral
characterization of entire operators given in \cite{II}. This
generalization is realized by extending the notion of entire operators
to a subclass of symmetric operators with deficiency indices $(1,1)$
that may have non-dense domain. The spectral characterization of a
given operator in the class is based on the distribution of the
spectra of its selfadjoint extensions within the Hilbert space. More
concretely, for a given simple, regular, closed symmetric (possibly
not densely defined) operator with deficiency indices $(1,1)$ to be
entire it is necessary and sufficient that the spectra of two of its
selfadjoint extensions satisfy conditions which reduce to the
convergence of certain series (the precise statement is
Proposition~\ref{prop:spectrum-tells-if-operator-is-entire}).

The class of entire operators was concocted by M. G. Krein as a tool
for treating in a unified way several classical problems in analysis
\cite{krein1,krein2,krein3,krein4}. The entire operators form a
subclass of the closed, densely defined, symmetric, regular operators
with equal deficiency indices. They have many remarkable properties as
is accounted for in the review book \cite{gorbachuk}. Krein's
definition of entire operators hinges on his functional model for
symmetric operators and it requires the existence of an element of the
Hilbert space with very peculiar properties. As first discussed in
\cite{II} it is possible to determine whether an operator is
entire by conditions that rely exclusively  on the distribution of
the spectra of selfadjoint extensions of the operator.

Although Krein's original work considers only densely defined
symmetric operators, it is clear that the definition of entire
operators can be extended to the case of not necessarily dense domain
with no formal changes (see Definition~\ref{def:entire-operators}).
Since non-densely and densely defined symmetric operators share
certain properties, the machinery developed in \cite{II} carries over
with some mild modifications.

One ingredient of our discussion is an extension of the functional
model developed in \cite{II}. This functional model associates a de
Branges space to every simple, regular, closed symmetric operator with
deficiency indices (1,1). It is worth remarking that functional models
for this and for related classes of operators have been implemented
before; see for instance \cite{debranges2,strauss}. However, the
functional model proposed in \cite{II} has shown to be particularly
suitable for us. Here we deem appropriate to mention \cite{martin} for
a related kind of results.

This paper is organized as follows. In Section 2 we recall some of the
properties held by operators that are closed, simple, symmetric
with deficiency indices $(1,1)$; the notion of entire operator is also
introduced here. Section 3 provides a short review on the theory of de
Branges Hilbert spaces, including those results relevant to this work,
in particular, a slightly modified version of a theorem due to Woracek
(Proposition~\ref{prop:1-in-dB-boosted}). In Section 4 we introduce a
functional model for any operator of the class under consideration so
that the model space is always a de Branges space. Finally, in Section 5 we
single out the class of de Branges spaces corresponding to entire
operators and  provide  necessary and sufficient
conditions on the spectra of two selfadjoint extensions of an entire
operator.
\begin{acknowledgments}
  Part of this work was done while the second author (J.\ H.\ T.)
  visited IIMAS--UNAM in January 2011. He sincerely thanks them for
  their kind hospitality.
\end{acknowledgments}

\section{On symmetric operators with not necessarily dense domain}
Let $\cH$ be a separable Hilbert space whose inner product
$\inner{\cdot}{\cdot}$ is assumed antilinear in its
first argument.
In this space we consider a closed, symmetric operator $A$ with
deficiency indices $(1,1)$. It is not assumed that its domain is dense
in $\cH$, therefore one should deal with the case when the adjoint of
$A$ is a linear relation. That is, in general,
\begin{equation}
\label{eq:adjoint-def}
A^* := \left\{\{\eta,\omega\}\in\cH\oplus\cH :
		\inner{\eta}{A\varphi}=\inner{\omega}{\varphi}
		\text{ for all }\varphi\in\dom(A)\right\}.
\end{equation}
Whenever the orthogonal complement of $\dom(A)$ is trivial, the set
$A^*(0):=\{\omega\in\cH:\{0,\omega\}\in A^*\}$ is also trivial,
i.\,e. $A^*(0)=\{0\}$, so $A^*$ is an operator; otherwise $A^*$ is a
proper closed linear relation.

For $z\in\C$ one has
\begin{equation}
\label{eq:adjoint-shifted}
A^*-zI := \left\{\{\eta,\omega-z\eta\}\in\cH\oplus\cH :
		\{\eta,\omega\}\in A^*\right\}
\end{equation}
so accordingly
\begin{equation}
\label{eq:ker-adjoint}
\Ker(A^*-zI) := \left\{\eta\in\cH : \{\eta,0\}\in A^*-zI \right\}.
\end{equation}
Since $\Ker(A^*-zI)=\cH\ominus\ran(A-\cc{z}I)$, our assumption on the
deficiency indices implies $\dim\Ker(A^*-zI)=1$ for all $z\in\C\setminus\R$.
Also, since
\begin{equation*}
A^*(0) = \left\{\omega\in\cH : \inner{\omega}{\psi} = 0
		\text{ for all }\psi\in\dom(A)\right\},
\end{equation*}
it is obvious that $A^*(0) = \dom(A)^\perp$.

The selfadjoint extensions within $\cH$ of a closed, non-densely
defined symmetric operator $A$ are the selfadjoint linear relations
that extend the graph of $A$. We recall that a linear relation $B$ is
selfadjoint if $B=B^*$ (as subsets of $\cH\oplus\cH$).

The following assertion follows easily from \cite[Section 1,
Lemma~2.2 and Theorem~2.4]{hassi1}.

\begin{proposition}
\label{prop:misc-about-symm-operators}
Let $A$ be a closed, non-densely defined, symmetric operator in $\cH$ with
deficiency indices $(1,1)$. Then:
\begin{enumerate}[{(i)}]
\item The codimension of $\dom(A)$ equals one.
\item All except one of the selfadjoint extensions of $A$ within $\cH$
	are operators.
\item Let $A_\gamma$ be one of the selfadjoint extensions of $A$ within $\cH$.
	Then the operator
	\begin{equation*}
	I + (z-w)(A_\gamma-zI)^{-1},\quad
        z\in\C\setminus\Sp(A_\gamma),\quad
        w\in\C
	\end{equation*}
	maps $\Ker(A^*-wI)$ injectively onto $\Ker(A^*-zI)$.
\end{enumerate}
\end{proposition}

In connection with this proposition we remind the reader that the
spectrum of a closed linear relation $B$ is the complement of the set
of all $z\in\C$ such that $(B-zI)^{-1}$ is a bounded operator defined
on all $\cH$.  Moreover, $\Sp(B)\subset\R$ when $B$ is a selfadjoint
linear relation \cite{dijksma}.

Given $\psi_{w_0}\in\Ker(A^*-w_0I)$, with $w_0\in\C\setminus\R$, let us define
\begin{equation}\label{eq:mapping-in-kernel}
\psi(z):= \left[I + (z-w_0)(A_\gamma-zI)^{-1}\right]\psi_{w_0},
\end{equation}
Note that $I + (z-w_0)(A_\gamma-zI)^{-1}$ is the generalized Cayley transform.
Obviously, $\psi(w_0)=\psi_{w_0}$. Moreover, a computation involving the
resolvent identity yields
\begin{equation}\label{eq:identity-between-kernels}
\psi(z) = \left[I + (z-v)(A_\gamma-zI)^{-1}\right]\psi(v),
\end{equation}
for any pair $z,v\in\C\setminus\R$. This identity will be used later on.

Let us now recall some concepts that will be used to single out
a class of closed symmetric operators with deficiency indices $(1,1)$.

A closed, symmetric operator $A$ is called {\em simple} if
\begin{equation*}
\bigcap_{z\in\C\setminus\R}\ran(A-zI) = \{0\}.
\end{equation*}
Equivalently, $A$ is simple if there exists no non-trivial subspace
$\cL\subset\cH$ that reduces $A$ and whose restriction to $\cL$ yields a
selfadjoint operator \cite[Proposition~1.1]{langer-textorius}.

There is one property specific to simple, closed symmetric operators with
deficiency indices $(1,1)$, that is of interest to us. It concerns their
commutativity with involutions. We say that an involution $J$ commutes with
a selfadjoint relation $B$ if
\[
J(B-zI)^{-1}\varphi = (B-\cc{z}I)^{-1}J\varphi,
\]
for every $\varphi\in\cH$ and $z\in\C\setminus\R$. If $B$ is moreover an
operator this is equivalent to the usual notion of commutativity, that is,
\[
J\dom(B)\subset\dom(B),\qquad JB\varphi = BJ\varphi
\]
for every $\varphi\in\dom(B)$.

\begin{proposition}
  Let $A$ be a simple, closed symmetric operator with deficiency
  indices $(1,1)$. Then there exists an involution $J$ that commutes
  with all its selfadjoint extensions within $\cH$.
\end{proposition}

\begin{proof}
  Choose a selfadjoint extension $A_\gamma$ and consider $\psi(z)$ as
  defined by \eqref{eq:mapping-in-kernel}. Recalling
  \eqref{eq:identity-between-kernels} along with the unitary character
  of the generalized Cayley transform, and applying the resolvent
  identity, one can verify that
\begin{equation}
\label{eq:pre-j}
\inner{\psi(\cc{z})}{\psi(\cc{v})} = \inner{\psi(v)}{\psi(z)}
\end{equation}
for every pair $z,v\in\C\setminus\R$.

Now define the action of $J$ on the set $\{\psi(z):z\in\C\setminus\R\}$ by
the rule
\[
J\psi(z)=\psi(\cc{z}),
\]
and on the set $\cD$ of finite linear combinations of such elements as
\[
J\left(\sum_{n}c_n\psi(z_n)\right): = \sum_{n}\cc{c_n}\psi(\cc{z_n}).
\]
Then, on one hand, \eqref{eq:pre-j} implies that $J$ is an involution on
$\cD$ which can be extended to all $\cH$ because of the simplicity of $A$.
On the other hand, since by the resolvent identity
\[
(A_\gamma-wI)^{-1}\psi(z) = \frac{\psi(z)-\psi(w)}{z-w},
\]
one obtains the identity
\[
J(A_\gamma-wI)^{-1}\psi(z)=(A_\gamma-\cc{w}I)^{-1}J\psi(z)
\]
which by linearity holds on $\cD$ and in turn it extends to all
$\cH$.

So far we know that $J$ commutes with $A_\gamma$. By resorting to the
well-known resolvent formula due to Krein (see \cite[Theorem
3.2]{hassi1} for a generalized formulation), one immediately obtains
the commutativity of $J$ with all the selfadjoint extensions of $A$
within $\cH$.
\end{proof}

A closed, symmetric operator is called {\em regular} if for every $z\in\C$
there exists $d_z>0$ such that
\begin{equation}
\label{eq:regular-point}
\norm{(A-zI)\psi}\ge d_z\norm{\psi},
\end{equation}
for all $\psi\in\dom(A)$. In other words, $A$ is regular if every point of
the complex plane is a point of regular type.

\begin{definition}
Let $\ournewclass$ be the class of simple, regular, closed symmetric
operators in $\cH$, whose deficiency indices are $(1,1)$.
\end{definition}

In \cite{I,II} we deal with the subclass of operators in
$\ournewclass$ that are densely defined. In the present work we extend
the results of \cite{II} to the larger class defined above.  At this
point it is convenient to touch upon some well-known properties shared
by the operators in $\ournewclass$ that are densely defined, and whose
generalizations to the whole class is rather straightforward. The
following statement is one of such generalizations which we believe
may have been already proven, however, due to the lack of the proper
reference, we provide the proof below.

\begin{proposition}
\label{prop:properties-of-new-class}
For $A\in\ournewclass$ the following assertions hold true:
\begin{enumerate}[{(i)}]
\item The spectrum of every selfadjoint extension of $A$ within $\cH$
	consists solely of isolated eigenvalues of multiplicity one.
\item Every real number is part of the spectrum of one, and only one,
	selfadjoint extension of $A$ within $\cH$.
\item The spectra of the selfadjoint extensions of $A$ within $\cH$ are
	pairwise interlaced.
\end{enumerate}
\end{proposition}
\begin{proof}
  Let us proof (i) in a way similar to the one used to prove
  \cite[Propositions 3.1 and 3.2]{gorbachuk}, but taking into account
  that the operator is not necessarily densely defined.

  For $A\in\ournewclass$ and any $r\in\R$ consider the constant $d_r$
  of (\ref{eq:regular-point}). Thus, the symmetric operator
  $(A-rI)^{-1}$, defined on the subspace $\ran(A-rI)$, is such that
  $\norm{(A-rI)^{-1}}\le d_r^{-1}$. By \cite[Theorem
  2]{krein-half-bounded} there is a selfadjoint extension $B$ of
  $(A-rI)^{-1}$ defined on the whole space and such that $\norm{B}\le
  d_r^{-1}$. Now, $B^{-1}$ is a selfadjoint extension of $A-rI$ and
  $\norm{B^{-1}f}\ge d_r\norm{f}$ for any $f\in\dom(B^{-1})$, which
  implies that the interval $(-d_r,d_r)\cap\Sp(B^{-1})=\emptyset$. By
  shifting $B^{-1}$ one obtains a selfadjoint extension of $A$ with no
  spectrum in the spectral lacuna $(r-d_r,r+d_r)$. By perturbation
  theory any selfadjoint extension of $A$ which is an operator has no
  points of the spectrum in this spectral lacuna other than one
  eigenvalue of multiplicity one. When $\cc{\dom(A)}\ne\cH$, the same
  is also true for the spectrum of the selfadjoint extension which is
  not an operator. This follows from a generalization of the
  Aronzajn-Krein formula (see \cite[Equation\,3.17]{hassi1}) after
  noting that the Weyl function is Herglotz and meromorphic for any
  selfadjoint extension being an operator. Now, for proving (i)
  consider any closed interval of $\R$, cover it with spectral lacunae
  and take a finite subcover.

  Once (i) has been proven, the assertions (ii) and (iii) follow from
  \cite[Equation\,3.17]{hassi1} and the properties of Herglotz
  meromorphic functions.
\end{proof}

\begin{definition}
\label{def:entire-operators}
  An operator $A\in\ournewclass$ is called {\em entire} if there
  exists $\mu\in\cH$ such that
\[
\cH = \ran(A-zI)\dot{+}\Span\{\mu\}
\]
for all $z\in\C$. Such $\mu$ is called an {\em entire gauge}.
\end{definition}

If $A\in\ournewclass$ turns out to be densely defined, then
Definition~\ref{def:entire-operators} reduces to Krein's
\cite[Section 1]{krein3}.  There are various densely defined operators
known to be entire \cite[Chapter 3]{gorbachuk},
\cite[Section 4]{krein3}. On the other hand, for what will be explained
in the subsequent sections, there are also entire operators with non-dense
domain.  Let us outline how one may construct an entire operator which
is not densely defined. The details of this construction will be
expounded in a further paper.

Consider the semi-infinite Jacobi matrix
\begin{equation}
  \label{eq:jm}
  \begin{pmatrix}
  q_1 & b_1 & 0  &  0  &  \cdots \\[1mm]
  b_1 & q_2 & b_2 & 0 & \cdots \\[1mm]
  0  &  b_2  & q_3  & b_3 &  \\
  0 & 0 & b_3 & q_4 & \ddots\\
  \vdots & \vdots &  & \ddots & \ddots
  \end{pmatrix}\,,
\end{equation}
where $b_k>0$ and $q_k\in\mathbb{R}$ for $k\in\mathbb{N}$. Fix an
orthonormal basis $\{\delta_k\}_{k\in\mathbb{N}}$ in $\cH$. Let $B$ be
the operator in $\cH$ whose matrix representation with respect to
$\{\delta_k\}_{k\in\mathbb{N}}$ is (\ref{eq:jm})
(cf. \cite[Section 47]{MR1255973}).  We assume that $B\ne B^*$,
equivalently, that $B$ has deficiency indices $(1,1)$ \cite[Chapter 4,
Section 1.2]{MR0184042}. Let $B_0$ be the restriction of $B$ to the set
$\{\phi\in\dom(B):\inner{\phi}{\delta_1}=0\}$. It follows from
(\ref{eq:adjoint-def}), (\ref{eq:adjoint-shifted}) and
(\ref{eq:ker-adjoint}) that $\eta\in\ker(B_0^*-zI)$ if and only if it
satisfies the equation
\begin{equation*}
  \inner{B\phi}{\eta}=\inner{\phi}{z\eta}\qquad\forall\phi\in\dom(B_0)\,.
\end{equation*}
Thus $\ker(B_0^*-zI)$ is the set of $\eta$'s in $\cH$ that satisfy
\begin{equation}
  \label{eq:difference}
  b_{k-1}\inner{\delta_{k-1}}{\eta}+q_k\inner{\delta_{k}}{\eta} +
  b_k\inner{\delta_{k+1}}{\eta}=z\inner{\delta_{k}}{\eta}\quad\forall k>1
\end{equation}
Hence $\dim\ker(B_0^*-zI)\le 2$. Now, let
\begin{equation*}
  \pi(z):=\sum_{k=1}^\infty P_{k-1}(z)\delta_k\qquad
  \theta(z):=\sum_{k=1}^\infty Q_{k-1}(z)\delta_k\,,
\end{equation*}
where $P_k(z)$, respectively $Q_k(z)$, is the $k$-th polynomial of
first, respectively second, kind associated to (\ref{eq:jm}). By the
definition of the polynomials $P_k(z)$ and $Q_k(z)$ \cite[Chapter 1,
Section 2.1]{MR0184042}, $\pi(z)$ and $\theta(z)$ are linearly
independent solutions of (\ref{eq:difference}) for every fixed
$z\in\C$. Moreover, since $B\ne B^*$, $\pi(z)$ and $\theta(z)$ are in
$\cH$ for all $z\in\C$ \cite[Theorems 1.3.1, 1.3.2]{MR0184042},
\cite[Theorem 3]{MR1627806}. So one arrives at the conclusion that,
for every fixed $z\in\C$,
\begin{equation*}
  \ker(B_0^*-zI)=\Span\{\pi(z),\theta(z)\}\,.
\end{equation*}
Any symmetric non-selfadjoint extension of $B_0$ has deficiency indices
(1,1). Furthermore, if $\kappa(z)$ is a ($z$-dependent) linear combination of 
$\pi(z)$ and $\theta(z)$ such that
$\inner{\kappa(z)}{\theta(z)}=0$ for all $z\in\C\setminus\R$, then (by a
parametrized version of \cite[Theorem 2.4]{MR1627806}) there
corresponds to an appropriately chosen isometry from
$\Span\{\kappa(z)\}$ onto $\Span\{\kappa(\cc{z})\}$ a non-selfadjoint
symmetric extension $\widetilde{B}$ of $B_0$ such that
$\dom(\widetilde{B})$ is not dense and
$\ker(\widetilde{B}^*-zI)=\Span\{\theta(z)\}$.  We claim that $\widetilde{B}$
is a non-densely defined entire operator. Indeed,
$\widetilde{B}\in\ournewclass$ (the simplicity follows from the
properties of the associated polynomials \cite[Chapter 1, Addenda and
Problems 7]{MR0184042}). Moreover, since
\begin{equation*}
  \inner{\theta(z)}{\delta_2}=b_1^{-1}\,,\qquad\forall z\in\C\,,
\end{equation*}
$\delta_2$ is an entire gauge.




\section{A review on de Branges spaces with zero-free functions}
\label{subsec:dB}
Let $\cB$ denote a nontrivial Hilbert space of entire functions with
inner product $\inner{\cdot}{\cdot}_\cB$. $\cB$ is a de Branges space
when, for every function $f(z)$ in $\cB$, the following conditions holds:
\begin{enumerate}[({A}1)]
\item For every $w\in\C\setminus\R$, the linear functional
        $f(\cdot)\mapsto f(w)$  is continuous;

\item for every non-real zero $w$ of $f(z)$, the function
        $f(z)(z-\cc{w})(z-w)^{-1}$ belongs to $\cB$
        and has the same norm as $f(z)$;

\item the function $f^\#(z):=\cc{f(\cc{z})}$ also belongs to $\cB$
        and has the same norm as $f(z)$.
\end{enumerate}

It follows from (A1) that for every non-real $w$ there is a function
$k(z,w)$ in $\cB$ such that $\inner{k(\cdot,w)}{f(\cdot)}_{\cB} =
f(w)$ for all $f(z)\in\cB$. Moreover,
$k(w,w)=\inner{k(\cdot,w)}{k(\cdot,w)}_\cB\ge 0$ where, as a
consequence of (A2), the positivity is strict for every non-real $w$
unless $\cB$ is $\C$; see the proof of Theorem~23 in
\cite{debranges}. Note that
$k(z,w)=\inner{k(\cdot,z)}{k(\cdot,w)}_\cB$ whenever $z$ and $w$ are
both non-real, therefore $k(w,z)=\cc{k(z,w)}$. Furthermore, due to
(A3) it can be shown that $\cc{k(\cc{z},w)}=k(z,\cc{w})$ for every
non-real $w$; we refer again to the proof of Theorem~23 in
\cite{debranges}.  Also note that $k(z,w)$ is entire with respect to
its first argument and, by (A3), it is anti-entire with respect to the
second one (once $k(z,w)$, as a function of its second argument, has
been extended to the whole complex plane \cite[Problem
52]{debranges}).

There is another way of defining a de Branges space. One starts by
considering an entire function $e(z)$ of the Hermite-Biehler class,
that is, an entire function without zeros in the upper half-plane
$\C^+$ that satisfies the inequality
$\abs{e(z)}>\abs{e^\#(z)}$ for $z\in\C^+$. Then, the de
Branges space $\cB(e)$ associated to $e(z)$
is the linear manifold of all entire functions $f(z)$ such that both
$f(z)/e(z)$ and $f^\#(z)/e(z)$ belong to the Hardy space
$H^2(\C^+)$, and equipped with the inner product
\[
\inner{f(\cdot)}{g(\cdot)}_{\cB(e)}:=
\int_{-\infty}^\infty\frac{\cc{f(x)}g(x)}{\abs{e(x)}^2}dx.
\]
It turns out that $\cB(e)$ is complete.

Both definitions of de Branges spaces are equivalent, viz., every
space $\cB(e)$ obeys (A1--A3); conversely, given a space $\cB$ there
exists an Hermite-Biehler function $e(z)$ such that $\cB$ coincides
with $\cB(e)$ as sets and the respective norms satisfy the equality
$\norm{f(\cdot)}_{\cB}=\norm{f(\cdot)}_{\cB(e)}$ \cite[Chapter
2]{debranges}. The function $e(z)$ is not unique; a choice for it is
\begin{equation*}
e(z)=-i\sqrt{\frac{\pi}{k(w_0,w_0)\im(w_0)}}
\left(z-\cc{w_0}\right)k(z,w_0),
\end{equation*}
where $w_0$ is some fixed complex number in  $\C^+$.

An entire function $g(z)$ is said to be associated to a de Branges
space $\cB$ if for every $f(z)\in\cB$ and
$w\in\C$,
\begin{equation*}
\frac{g(z)f(w)-g(w)f(z)}{z-w}\in\mathcal{B}.
\end{equation*}
The set of associated functions is denoted $\assoc\mathcal{B}$.  It
is well known that
\begin{equation*}
\assoc\mathcal{B} = \mathcal{B} + z\mathcal{B};
\end{equation*}
see \cite[Theorem 25]{debranges} and \cite[Lemma 4.5]{kaltenback} for
alternative characterizations. In passing, let us note that
$e(z)\in\assoc\mathcal{B}(e)\setminus\mathcal{B}(e)$; this fact
follows easily from \cite[Theorem 25]{debranges}.

The space $\assoc\mathcal{B}(e)$ contains a distinctive family of
entire functions. They are given by
\begin{equation*}
s_\beta(z):=\frac{i}{2}\left[e^{i\beta}e(z)-e^{-i\beta}e^\#(z)\right],
	\quad \beta\in[0,\pi).
\end{equation*}
These real entire functions are related to the selfadjoint extensions
of the multiplication operator $S$ defined
by \begin{equation}\label{eq:multiplication-operator}
  \dom(S):=\{f(z)\in\mathcal{B}: zf(z)\in\mathcal{B}\},\quad
  (Sf)(z)=zf(z).
\end{equation}
This is a simple, regular, closed symmetric operator with deficiency
indices $(1,1)$ which is not necessarily densely defined
\cite[Proposition 4.2, Corollary 4.3, Corollary 4.7]{kaltenback}.  It
turns out that $\cc{\dom(S)}\neq\mathcal{B}$ if and only if there
exists $\gamma\in[0,\pi)$ such that
$s_\gamma(z)\in\mathcal{B}$. Furthermore,
$\dom(S)^\perp=\Span\{s_\gamma(z)\}$ \cite[Theorem~29]{debranges} and
\cite[Corollary~6.3]{kaltenback}; compare with (i) of
Proposition~\ref{prop:misc-about-symm-operators}.

For any selfadjoint extension $S_\sharp$ of $S$ there exists a
unique $\beta$ in $[0,\pi)$ such that
\begin{equation}
\label{eq:selfadjoint-extensions-s}
(S_\sharp-wI)^{-1}f(z)
	= \frac{f(z)-\frac{s_\beta(z)}{s_\beta(w)}f(w)}{z-w},\quad
		w\in\C\setminus\Sp(S_\sharp),\quad f(z)\in\cB.
\end{equation}
Moreover, $\Sp(S_\sharp)=\left\{x\in\mathbb{R}: s_\beta(x)=0\right\}$.
\cite[Propositions~4.6 and 6.1]{kaltenback}. If $S_\sharp$ is a selfadjoint
operator extension of $S$, then \eqref{eq:selfadjoint-extensions-s} is
equivalent to
\begin{gather*}
\dom(S_\sharp) =
	\left\{g(z)=\frac{f(z)-
\frac{s_\beta(z)}{s_\beta(z_0)}f(z_0)}{z-z_0},
	\quad f(z)\in\mathcal{B},\quad z_0:s_\beta(z_0)\neq 0\right\},
\\[2mm]
(S_\sharp g)(z) = z g(z) +
\frac{s_\beta(z)}{s_\beta(z_0)}f(z_0).\nonumber
\end{gather*}
The eigenfunction $g_x$ corresponding to $x\in\Sp(S_\sharp)$ is given
(up to normalization) by
\[
g_x(z)=\frac{s_\beta(z)}{z-x}.
\]
Thus, since $S$ is regular and simple, every $s_\beta(z)$ has only real
zeros of multiplicity one and the (sets of) zeros of any pair $s_\beta(z)$ and
$s_{\beta'}(z)$ are always interlaced.

The proof of the following result can be found in \cite{woracek2} for
a particular pair of selfadjoint extensions of $S$. Another proof, when the 
operator $S$ is densely defined, is given in \cite[Proposition 3.9]{II}.

\begin{proposition}\label{prop:1-in-dB-boosted}
  Suppose $e(x)\neq 0$ for $x\in\mathbb{R}$ and
  $e(0)=(\sin\gamma)^{-1}$ for some fixed $\gamma\in(0,\pi)$. Let
  $\{x_n\}_{n\in\mathbb{N}}$ be the sequence of zeros of the function
  $s_\gamma(z)$. Also, let $\{x_n^+\}_{n\in\mathbb{N}}$ and
  $\{x_n^-\}_{n\in\mathbb{N}}$ be the sequences of positive,
  respectively negative, zeros of $s_\gamma(z)$, arranged according
  to increasing modulus.  Then a zero-free, real entire function
  belongs to $\mathcal{B}(e)$ if and only if the following
  conditions hold true:
\begin{enumerate}[(C1)]
\item The limit
	$\displaystyle{\lim_{r\to\infty}\sum_{0<|x_n|\le r}
		\frac{1}{x_n}}$
	exists;
\item $\displaystyle{\lim_{n\to\infty}\frac{n}{x_n^{+}}
		=- \lim_{n\to\infty}\frac{n}{x_n^{-}}<\infty}$;
\item Assuming that $\{b_n\}_{n\in\mathbb{N}}$ are the zeros of
  $s_\beta(z)$, define
	\[
	h_\beta(z):=\left\{\begin{array}{ll}
			\displaystyle{\lim_{r\to\infty}\prod_{|b_n|\le r}
			\left(1-\frac{z}{b_n}\right)}
				& \text{ if 0 is not a root of } s_\beta(z),
			\\
			\displaystyle{z\lim_{r\to\infty}\prod_{0<|b_n|\le r}
			\left(1-\frac{z}{b_n}\right)}
				& \text{ otherwise. }
			   \end{array}\right.
	\]
	The series
	$\displaystyle{
		\sum_{n\in\mathbb{N}}\abs{\frac{1}
		{h_{0}(x_n)h_{\gamma}'(x_n)}}}$ is convergent.
\end{enumerate}
\end{proposition}
\begin{proof}
Combine Theorem~3.2 of \cite{woracek2} with Lemmas~3.3 and 3.4 of \cite{II}.
\end{proof}

\section{A functional model for operators in $\ournewclass$}
The functional model given in this section follows the construction
developed in \cite{II}, now adapted to include all the operators in the
class $\ournewclass$. This functional model is based on (the properties of)
the operator mentioned in (iii) of
Proposition~\ref{prop:misc-about-symm-operators} with the following addition.

\begin{proposition}
Given $A\in\ournewclass$, let $J$ be an involution that commutes with
one of its selfadjoint extensions within $\cH$ (hence with all of them),
say, $A_\gamma$. Choose $v\in\Sp(A_\gamma)$. Then, there exists
$\psi_v\in\Ker(A^*-vI)$ such that $J\psi_v=\psi_v$.
\end{proposition}
\begin{proof}
  Let $\phi_v$ be an element of $\Ker(A_\gamma-vI)$. Since $J$
  commutes with $A_\gamma$, one immediately obtains that
  $J\phi_v\in\Ker(A_\gamma-vI)$. But, by our assumption on the
  deficiency indices of $A$ and its regularity, $\Ker(A^*-vI)$ is a
  one dimensional space and it contains $\Ker(A_\gamma-vI)$. So, in
  $\Ker(A_\gamma-vI)$, $J$ reduces to multiplication by a scalar
  $\alpha$ and the properties of the involution imply that
  $\abs{\alpha}=1$. Now, $\psi_v:= (1+\alpha)\phi_v$ has the required
  properties.
\end{proof}

Given $A\in\ournewclass$ and an involution $J$ that commutes
with its selfadjoint extensions within $\cH$, define
\begin{equation}
\label{eq:xi-def}
\xi_{\gamma,v}(z)
	:=h_\gamma(z)\left[I+(z-v)(A_\gamma-zI)^{-1}\right]\psi_v\,,
\end{equation}
where $v$ and $\psi_v$ are chosen as in the previous proposition, and
$h_\gamma(z)$ is a real entire function whose zero set is
$\Sp(A_\gamma)$ (see Proposition \ref{prop:properties-of-new-class}
(i)). Clearly, up to a zero-free real entire function,
$\xi_{\gamma,v}(z)$ is completely determined by the choice of the
selfadjoint extension $A_\gamma$ and $v$. Actually, as it is stated
more precisely below, $\xi_{\gamma,v}(z)$ does not depend on
$A_\gamma$ nor on $v$.
\begin{proposition}
\label{prop:xi-properties}
\begin{enumerate}[(i)]
\item The vector-valued function $\xi_{\gamma,v}(z)$ is zero-free and
	entire. It lies in $\Ker(A^* - zI)$ for every $z\in\C$.
\item $J\xi_{\gamma,v}(z)=\xi_{\gamma,v}(\cc{z})$ for all $z\in\C$.
\item Given $\xi_{\gamma_1,v_1}(z)$ and
  $\xi_{\gamma_2,v_2}(z)$, there exists a zero-free real entire
  function $g(z)$ such that $\xi_{\gamma_2,v_2}(z)=
  g(z)\xi_{\gamma_1,v_1}(z)$.
\end{enumerate}
\end{proposition}
\begin{proof}
  Due to  (iii) of
  Proposition~\ref{prop:misc-about-symm-operators}, the proof of (i)
  is rather straightforward. In fact, one should only follow the
  first part of the proof of \cite[Lemma 4.1]{II}. The proof of (ii)
  also follows easily from our choice of $\psi_w$ and $h_\gamma(z)$ in
  the definition of $\xi_{\gamma,w}(z)$. To prove (iii), one first uses
  (iii) of
  Proposition~\ref{prop:misc-about-symm-operators} and the fact that
  $\dim\Ker(A^*-wI)=1$ to obtain that $\xi_{\gamma_2,w_2}(z)$ and
  $\xi_{\gamma_1,w_1}(z)$ differ by a nonzero scalar complex
  function. Then the reality of this function follows from (ii).
\end{proof}

For the reason already explained, from now on the
function $\xi_{\gamma,v}(z)$ will be denoted by $\xi(z)$. Now define
\begin{equation*}
\left(\Phi\varphi\right)(z):=\inner{\xi(\cc{z})}{\varphi},\qquad
	\varphi\in\cH.
\end{equation*}
$\Phi$ maps $\cH$ onto a certain linear manifold $\widehat{\cH}$ of
entire functions. Since $A$ is simple, it follows that $\Phi$ is
injective.  A generic element of $\widehat{\cH}$ will be denoted
by $\widehat{\varphi}(z)$, as a reminder of the fact that it is the
image under $\Phi$ of a unique element $\varphi\in\cH$.

The linear space $\widehat{\cH}$ is turned into a Hilbert space by defining
\begin{equation*}
  \inner{\widehat{\eta}(\cdot)}{\widehat{\varphi}(\cdot)}:=
\inner{\eta}{\varphi}\,.
\end{equation*}
Clearly, $\Phi$ is an isometry from $\cH$ onto $\widehat{\cH}$.
\begin{proposition}
$\hat{\cH}$ is a de Branges space.
\end{proposition}
\begin{proof}
It suffices to show that the axioms given at the beginning of
Section~\ref{subsec:dB} holds for $\hat{\cH}$.

It is straightforward to verify that $k(z,w):=\inner{\xi(\cc{z})}{\xi(\cc{w})}$
is a reproducing kernel for $\hat{\cH}$. This accounts for (A1).

Suppose $\hat{\varphi}(z)\in\hat{\cH}$ has a zero at $z=w$. Then its preimage
$\varphi\in\cH$ lies in $\ran(A-wI)$. This allows one
to set $\eta\in\cH$ by
\[
\eta  = (A-\cc{w}I)(A-wI)^{-1}\varphi
	= \varphi + (w-\cc{w})(A_\gamma-wI)^{-1}\varphi.
\]
Now, recalling \eqref{eq:xi-def} and applying the resolvent identity one
obtains
\[
\inner{\xi(\cc{z})}{\eta} = \frac{z-\cc{w}}{z-w}\inner{\xi(\cc{z})}{\varphi}.
\]
Since $\eta$ and $\varphi$ are related by a Cayley transform, the equality
of norms follows. This proves (A2).

As for (A3), consider any $\hat{\varphi}(z)=\inner{\xi(\cc{z})}{\varphi}$.
Then, as a consequence of (ii) of Proposition~\ref{prop:xi-properties},
one has $\hat{\varphi}^\#(z)=\inner{\xi(\cc{z})}{J\varphi}$.
\end{proof}

It is worth remarking that the last part of the proof given above
shows that $^\#=\Phi J \Phi^{-1}$.

The following obvious assertion is the key of (every) functional
model; we state it for the sake of completeness.

\begin{proposition}
Let $S$ be the multiplication operator on
  $\widehat{\cH}$ given by (\ref{eq:multiplication-operator}).
\begin{enumerate}[{(i)}]
\item $S=\Phi A\Phi^{-1}$ and $\dom(S)=\Phi\dom(A)$.
\item The selfadjoint extensions of $S$ within $\hat{\cH}$ are in one-one
	correspondence with the selfadjoint extensions of $A$ within $\cH$.
\end{enumerate}
\end{proposition}

Item (ii) above can be stated more succinctly by saying that
\[
\Phi(A_\beta - zI)^{-1}\Phi^{-1} = (S_\beta - zI)^{-1},\quad
z\in\C\setminus\Sp(A_\gamma),
\]
for all $\beta$ of a certain (common) parametrization of the selfadjoint
extensions of both $A$ and $S$. 
This expression is of course valid even for the exceptional (i.e. 
non-operator) selfadjoint extension of $A$. In passing we note that
the exceptional selfadjoint extension of a non-densely defined operator 
in $\ournewclass$ corresponds to the selfadjoint extension of the 
operator $S$ whose associated function lies in $\widehat{\cH}$.

\section{Spectral characterization}
In the previous section we constructed a functional model that
associates a de Branges space to every operator $A$ in $\ournewclass$
in such a way that the operator of multiplication in the de Branges
space is unitarily equivalent to $A$. The first task in this section
is to single out the class of de Branges spaces corresponding to
entire operators in our functional model. Having found this class, we
use the theory of de Branges spaces to give a spectral
characterization of the multiplication operator for the class we
found. This is how we give necessary and sufficient conditions on the
spectra of two selfadjoint extensions of an entire operator.

The following proposition gives a characterization of the class of
de Branges spaces corresponding to entire operators in our functional
model.
\begin{proposition}
  \label{prop:when-the-operator-is-entire}
  $A\in\ournewclass$ is entire if and only if $\widehat{\cH}$ contains
  a zero-free entire function.
\end{proposition}
\begin{proof}
  Let $g(z)\in\widehat{\cH}$ be the function whose
existence is assumed. Clearly there exists (a unique) $\mu\in\cH$ such
that $g(z)\equiv\inner{\xi(\cc{z})}{\mu}$. Therefore, $\mu$ is never
orthogonal to $\Ker(A^*-zI)$ for all $z\in\mathbb{C}$.
That is, $\mu$ is an entire gauge for the operator $A$.

The necessity is established by noting that the image of the entire
gauge under $\Phi$ is a zero-free function.
\end{proof}
\begin{proposition}
\label{prop:spectrum-tells-if-operator-is-entire}
For $A\in\ournewclass$, consider the selfadjoint extensions (within
$\cH$) $A_0$ and $A_{\gamma}$, with $0<\gamma<\pi$. Then $A$ is entire
with real entire gauge $\mu$ ($J\mu=\mu$) if and only if $\Sp(A_0)$
and $\Sp(A_{\gamma})$ obey conditions (C1), (C2) and (C3) of
Proposition~\ref{prop:1-in-dB-boosted}.
\end{proposition}
\begin{proof}
Apply Proposition~\ref{prop:1-in-dB-boosted} along with
Proposition~\ref{prop:when-the-operator-is-entire}.
\end{proof}

We remark that when $A$ is an entire operator with non-dense domain, it
may be that either $A_0$ or $A_\gamma$ is not an operator (see
Proposition \ref{prop:misc-about-symm-operators} (ii)). Nevertheless,
even in this case, $\Sp(A_0)$ and $\Sp(A_{\gamma})$ satisfy (C1), (C2)
and (C3).

The following proposition shows, among other things, that the original
functional model by Krein is a particular case of our functional model.
\begin{proposition}
Assume $1\in\widehat{\cH}$. Then there exists $\mu\in\cH$
such that
\[
h_\gamma(z)=\inner{\psi_v + (z-v)(A_\gamma-zI)^{-1}\psi_v}{\mu}^{-1}
\]
and $J\mu=\mu$. Moreover, $\mu$ is the unique entire gauge of $A$ modulo
a real scalar factor.
\end{proposition}
\begin{proof}
  Necessarily, $1\equiv\inner{\xi(\cc{z})}{\mu}$ for some
  $\mu\in\cH$. By (\ref{eq:xi-def}), and taking into account the
  occurrence of $J$, one obtains the stated expression for
  $h_\gamma(z)$.  By the same token, the reality of $\mu$ is shown.

  Suppose that there are two real entire gauges $\mu$ and $\mu'$.
  The discussion in Paragraph~5.2 of \cite{gorbachuk} shows that
  $(\Phi_{\mu}\mu')(z)=ae^{ibz}$ with $a\in\mathbb{C}$ and
  $b\in\mathbb{R}$.  Due to the assumed reality, one concludes that
  $b=0$ and $a\in\mathbb{R}$.
\end{proof}

\section{Concluding remarks}

We would like to add some few comments concerning further extensions
of the present work.

First, since there are de Branges spaces that contain the constant
functions but whose multiplication operator is not densely defined,
it follows that, apart from the example given in Section~2, there
should be other operators in the class introduced in this work
that are not comprised in the original Krein's notion of entire
operators. The details of our example as well as other ones and
applications of our results will be studied elsewhere.

Second, it is possible to define a notion of a (possibly non-densely
defined) operator that is entire in a generalized sense, much in the
same vein as the original definition by Krein for densely defined
operators (see \cite[Chapter 2, Section 9]{gorbachuk}). Following
\cite[Section 5]{II}, operators entire in this generalized sense could
also be characterized by the spectra of their selfadjoint extensions.

Finally, it is known that the set of selfadjoint operator extensions
within $\cH$ of a non-densely defined operator are in one-one
correspondence with a set of rank-one perturbations of one of these
selfadjoint operator extensions \cite[Section 2]{hassi1}.  This set of
rank-one perturbations is generated by elements in $\cH$ so it seems
interesting to study the relation (if any) between these elements and
the gauges of operators in $\ournewclass$. Ultimately, we believe that
a suitable characterization of the rank-one perturbations could
provide another necessary and sufficient condition for a non-densely
defined operator in $\ournewclass$ to be entire. This problem, as well
as the previous one, will be discussed in a subsequent work.

\end{document}